\PassOptionsToPackage{square,comma,numbers,sort&compress}{natbib}
\documentclass{article}       
\usepackage{arxiv} 

\usepackage[utf8]{inputenc} 
\usepackage[T1]{fontenc}    

\usepackage{graphicx}
\usepackage{amsthm,amsmath,natbib,color, amsfonts, mathtools,amssymb,commath}
\usepackage{tikz}
\usetikzlibrary{automata, positioning}
\usetikzlibrary{cd}
\usepackage{hyperref}
\usepackage{cleveref}

\newtheorem{theorem}{Theorem}[section]

\newtheorem{lemma}[theorem]{Lemma}
\newtheorem{remark}[theorem]{Remark}
\newtheorem{definition}[theorem]{Definition}

\newtheorem{example}[theorem]{Example}

\newcommand{\sgn}{\operatorname{sgn}}

\title{Cyclic attractors of nonexpanding $n$-ary networks}

\author{
			Etan Basser-Ravitz \\
			Department of Mathematics\\
			Yale University\\
			New Haven, CT, USA\\
			 \texttt{etan.basser@yale.edu}
			\And 
		   Arman Darbar \\
		   New York University\\
		   New York, NY, USA\\
		    \texttt{agd8612@nyu.edu}
		   	\AND
		   Julia Chifman\thanks{Corresponding author}\\
   			Department of Mathematics and Statistics \\
   			American University \\
   			Washington DC, USA\\
   			\texttt{chifman@american.edu} \\
}

\date{}
\renewcommand{\headeright}{}
\renewcommand{\undertitle}{}

\begin{document}
\maketitle

\begin{abstract}
Discrete dynamical systems in which model components take on categorical values have been successfully applied to biological networks to study their global dynamic behavior. Boolean models in particular have been used extensively. However, multi-state models have also emerged as effective computational tools for the analysis of complex mechanisms underlying biological networks. Models in which variables assume more than two discrete states provide greater resolution, but this scheme introduces discontinuities. In particular, variables can increase or decrease by more than one unit in one time step. This can be corrected, without changing fixed points of the system, by applying an additional rule to each local activation function. On the other hand, if one is interested in cyclic attractors of their system, then this rule can potentially introduce new cyclic attractors that were not observed previously. This article makes some advancements in understanding the state space dynamics of multi-state network models with a synchronous update schedule and establishes conditions under which no new cyclic attractors are added to networks when the additional rule is applied.  
Our analytical results have the potential to be incorporated into modeling software and aid researchers in their analyses of biological multi-state networks.      
\end{abstract}

\section{Introduction}
\label{sec:introduction}

Biological networks can be modeled through a variety of mathematical formalisms and the choice depends on the type of questions and data available. Quantitative models require detailed knowledge of the kinetics involved, which at times is very limited. Thus, many biological systems are modeled by discrete dynamical systems, qualitative models in which variables take on categorical values, such as {\em on/off} or {\em low/medium/high}. Much of literature and computational tools are devoted to Boolean models, however, not all biological networks are Boolean.  
The interest in models where variables take on more than two values can be traced back to R. Thomas \cite{THOMAS19911}. Examples of such biological models are: (i) multi-state models \cite{THIEFFRY1995277,SANCHEZ2001115}, (ii) models that employ a combination of Boolean and ternary variables \cite{espinosa2004,Remy4042,ABOUJAOUDE2009561}, (iii) strictly ternary models \cite{chifman2017} and (iv) four state models \cite{Setty7702}. Networks of large biosystems in which variables assume more than two discrete states do present a computational challenge, but with the development of faster computational tools, it is only plausible to assume that the scientific community will see a rise in the use of multi-state models. The aim of this article is to provide analytical results that have the potential to be incorporated into a software, including our own tool \url{https://steadycellphenotype.github.io} \cite{Knapp2021}, allowing the user to assess the properties of their multi-state biological network prior to analysis. 

This work was motivated by our extensive analyses of biological networks. In particular, the focus of our team is on intracellular iron metabolism in epithelial cells and its dysregulation in cancer \cite{chifman2012,chifman2017}. We have also explored  iron handling in macrophages \cite{Vieira2442}. Our current model includes multiple oncogenic pathways, iron homeostasis and utilization pathways, oxidative stress response and a cell G1/S phase component, totaling 64 species. With the exception of an ODE model  \cite{chifman2012}, we chose ternary logic for all models since several species in our network, including iron levels, could not be modeled as On/Off or Active/Inactive.  Allowing three states provided us with much greater resolution,  but this scheme also introduced discontinuities, meaning that variables can jump $\pm2$ units in one time step.   If more states are required for each species, these jumps become even greater.  Many multi-state models have been constructed under the constraint that species change by at most one unit (e.g., \cite{ABOUJAOUDE2009561,chifman2017}). This can be achieved by applying an additional  rule to each local activation function (Section~\ref{sec:one-step}). Intuitively, this rule takes into consideration the current state of the node at time $t$ and its next value at time $t+1$.  For example, in the ternary case,  if the current state of the node is {\em low} (0),  but a local update function for the same node assigns  a  {\em high} (2) value, then the actual value this node would receive at time $t+1$  will be {\em normal} (1).  Applying this rule does not change fixed points (Remark~\ref{rm:fixed_pt}), yet, cyclic attractors of ternary networks can present a challenge. 

For clarity, let $\mathcal{N}$ represent a ternary network in which $\pm 2$ jumps are present, and let $\mathcal{H}$ be the corresponding network in which the additional rule was added to all local activation functions of $\mathcal{N}$ to ensure that variables change by at most one unit. We have simulated and analyzed many networks under a synchronous update schedule and found that some of the networks had rather intuitive connections between cyclic attractors: (i) cyclic attractors that had $\pm 2$ jumps in $\mathcal{N}$ were absorbed in $\mathcal{H}$ by the basins of fixed points or by the basins of other remaining cyclic attractors, (ii) cyclic attractors that had only $\pm 1$ jumps remained in $\mathcal{H}$ and (iii) {\em no new} cyclic attractors were introduced. However, we have also noticed that many networks, even those with the same interaction graph and nearly identical local activation functions, had entirely {\em new} cyclic attractors created in $\mathcal{H}$ (see Example~\ref{ex:networks} and Figures~\ref{fig:1}(d)-(e)). This observation prompted our team to investigate this phenomenon further.  Additionally, we have noticed that the state space of the network $\mathcal{H}$ without new cyclic attractors preserved some rigid structure: trajectories of the networks $\mathcal{N}$ and $\mathcal{H}$ merged after some iteration, which was not true for the networks that added additional cyclic attractors (examples can be found at the beginning of the Section~\ref{sec:nonexpanding}). Other differences in the state space dynamics between $\mathcal{N}$ and $\mathcal{H}$ networks are mentioned in the  Discussion Section~\ref{sec:disc}.

Attractors of discrete dynamical systems are of particular importance since in the biological context they correspond to different phenotypes, a concept that appeared in seminal works by S.A. Kauffman, L. Glass and S. Huang \cite{KAUFFMAN1969437,GLASS1973103,sui1999}. Deciding when an attractor (cyclic or fixed point) is biologically relevant has been addressed by several authors ranging from the idea of stability \cite{klemm2005}, the size of the basin of attraction \cite{Bornholdt2008BooleanNM,10.1371/journal.pone.0195126} to ergodic attractors \cite{RIBEIRO2007743}. We encourage the reader to consult a concise review by J. D. Schwab et al. \cite{SCHWAB2020571} about Boolean models. From our own experience we have found that both fixed points and cyclic attractors can be of biological interest.  While working with ternary models, we have learned that having a network that does not introduce new cyclic attractors after imposing $\pm1$ jumps produced biologically relevant dynamics, at least in our case of iron metabolism. Whether or not it is true in general will require further investigations, but  to our knowledge, the correspondence between trajectories and cyclic attractors of the biological multi-state networks with and without jumps has not been rigorously addressed. {\em This article makes the first step in describing some conditions under which no new cyclic attractors are added to networks when $\pm 1$ jumps are enforced under a synchronous update schedule}. The results in this article  make a contribution not only to the analysis of biological multi-state networks but also to discrete applied mathematics in general.      

The article is organized as follows. Section~\ref{sec:prelims} provides some general terminology and definitions of $n$-ary networks, including the ones in which $\pm1$ jumps are enforced. We call these networks $one$-step networks. Additionally, this section singles out one class of biological networks named in this article as generalized Boolean $n$~\text{-ary} networks.  Section~\ref{sec:nonexpanding} defines a new class of networks: {\em nonexpanding} networks with respect to the Chebyshev distance. Boolean {\em non-expansive} networks under the Hamming distance were analyzed in \cite{RICHARD20111085,RICHARD20151} with a focus on fixed points. In this article, our main focus is on cyclic attractors, and also on the differences and the similarities between $n$-ary networks and their corresponding $one$-step networks.  In the same section, we then prove our first result that {\em generalized Boolean $n$-ary networks are nonexapnding} (Theorem~\ref{thm:gen_boolean}).  We also discuss our observation about trajectories and prove our second result (Theorem~\ref{thm:composition}) about the interplay between global activation functions of nonexpanding networks and their corresponding $one$-step networks. This then leads to our main result that {\em corresponding  $one$-step networks of nonexpanding networks do not add new cyclic attractors} (Thorem~\ref{thm:no_cycles}).  
We conclude the article with the application of our results to intracellular iron networks (Section~\ref{sec:application}) and with the Discussion (Section~\ref{sec:disc}).

\section{Preliminaries}\label{sec:prelims}
\label{sec:networks}

We generalize the concept of a {\em Boolean Network} on $m$ components (species) to networks in which each node of a network has an associated state in $\{0,1,2\dots,n-1\}$. Some notation in this section concerning graphs is standard (see for example \cite{west_introduction_2000,chartrand1993applied}).

\subsection{Digraphs}
\label{section:graphs}
Let $G~=~(V(G), E(G))$ be a finite digraph with a vertex set and an arc set denoted by $V(G)$ and $E(G) $, respectively. An arc  joining two vertices  $v, u~\in~V(G)$ {\em directed} from $u$ to $v$ will be denoted by $(u,v)$. The first vertex $u$ is called the {\em tail (source)} and the second vertex $v$ is the {\em head (target)} of the arc. The number of incoming arcs of a vertex $v~\in~V(G)$ is called the {\em in~-~degree} of $v$ and the number of outgoing arcs is called the {\em out~-~degree} of $v$. For a vertex $v~\in~V(G)$ the {\em in~-~neighborhood} (predecessor set) and {\em out~-~neighborhood} (successor set) of $v$ are defined by $N^{-}(v)~=~\{u~\in~V(G)~|~(u, v)~\in~E(G)\}$ and $N^+(v)~= ~\{u~\in~V(G)~|~(v,u)~\in~E(G)\}$, respectively.

An arc that joins a vertex to itself is called a {\em loop}, and two or more arcs that join the same pair of vertices are called {\em multiple (parallel) edges}.  A graph that has no loops or multiple  edges is a {\em simple graph}.  A graph $G$ could be a signed digraph, meaning there is a sign function $\sigma: E(G) \to \{+,-\}$. In this case, each arc receives either a positive $(+)$ or negative $(-)$ assignment. In a biological context, such arcs are often called {\em activation} (positive arc) or {\em inhibition} (negative arc). 

A {\em directed path} is a sequence of distinct vertices $v_1, v_2, \dots, v_k$ such that there is an arc $(v_i, v_{i+1})$ for all $i \in \{1,2, \dots, k-1\}$. If, in addition, there is an arc from $v_k$ to $v_1$ then a directed path is called a {\em cycle}. Note, directed paths and cycles are simple subgraphs. 
A digraph $G$ is {\em strongly connected or strong} if for any $v, u \in V(G)$ with $v \neq u$ there is a directed path from $u$ to $v$ and from $v$ to $u$. 

\subsection{$n$-ary Networks}
 Let $X_n:=\{0,1,2,\dots,n-1\}$ be a finite set on $n$ elements equipped with a linear order.  In the biological context the elements of $X_n$ can be associated with concentration/activity levels of biological species. 
 
An {\em $n$-ary network}  $\mathcal{N} = (G,F)$ with $n\geq 2$  is a discrete dynamical system where $G= (V(G), E(G))$ is a digraph on $m$ nodes, called an {\em interaction graph}, and $F$ is a global transition function 
 \[
 F = (f_1, f_2, \dots, f_m):X_n^m \to X_n^m.
\]
Each node $v \in V(G)$ has an associated state in $X_n$ and the value of $v$ is updated by a coordinate {\em local activation function} $f_v:X_n^m \to X_n$ whose value depends on the values of the nodes in $N^-(v)$. Note, with this definition, if $n =2$ and each local activation function is a Boolean function, then $\mathcal{N} = (G, F)$ is a Boolean network. 

In this article we assume a {\em synchronous} update, meaning all nodes are updated simultaneously based on the values of their input/incident nodes  at the previous time step and a local activation function. Thus, the dynamics of the network $\mathcal{N}$ is described by the successive iterations of $F$ and each state $\mathbf{x} = (x_1,x_2,\dots, x_m) \in X_n^m$ leads to another state, eventually converging to a fixed point or a cyclic attractor of length $k$, also termed a $k$-cycle (see diagram below). 

\begin{small}
\begin{tikzcd}
 fixed \ point & \mathbf{x} \arrow[r] 
	& F(\mathbf{x}) \arrow[r] 
		& F^2(\mathbf{x})  \arrow[r] 
			& \cdots  \arrow[r] 
				& F^t(\mathbf{x}) \arrow[loop right]  
\end{tikzcd}
			
\begin{tikzcd}
	 k-cycle \hspace{0.12in} & \mathbf{x} \arrow[r] 
	& F(\mathbf{x}) \arrow[r] 
	& \cdots  \arrow[r] 
	& F^{i}(\mathbf{x})  \arrow[r] 
	& \cdots  \arrow[r] 
	& F^{i+j}(\mathbf{x}) \arrow[d] &\\	
	 & & & &F^{i+(k-1)}(\mathbf{x})  \arrow[u] &\cdots \arrow[l] & F^{i+j+1}(\mathbf{x}) \arrow[l]
\end{tikzcd}
\end{small}	

\noindent In the above diagram, $F^i(\mathbf{x})$ means a composition of $F$ with itself $i$ times. Since the interaction graph is finite we will have $n^m$ states and each state $\mathbf{x} \in X_n^m$ will belong to the basin of attraction of only one attractor, a fixed point or a cycle attractor. 

We conclude this subsection with three examples that have the same interaction graph but different global transition functions. We will use these examples throughout the text to demonstrate our ideas. We would like to emphasize that our results in this article hold for graphs with multiple arcs and/or loops, and do not require the interaction graph to be signed.

\newpage
\begin{example}\label{ex:networks}
Let $G$ be a digraph as depicted in Figure~\ref{fig:1}(a) and let $X_3 = \{0,1,2\}$.
\begin{figure}[ht!]
\centering
\includegraphics[width=0.8\textwidth]{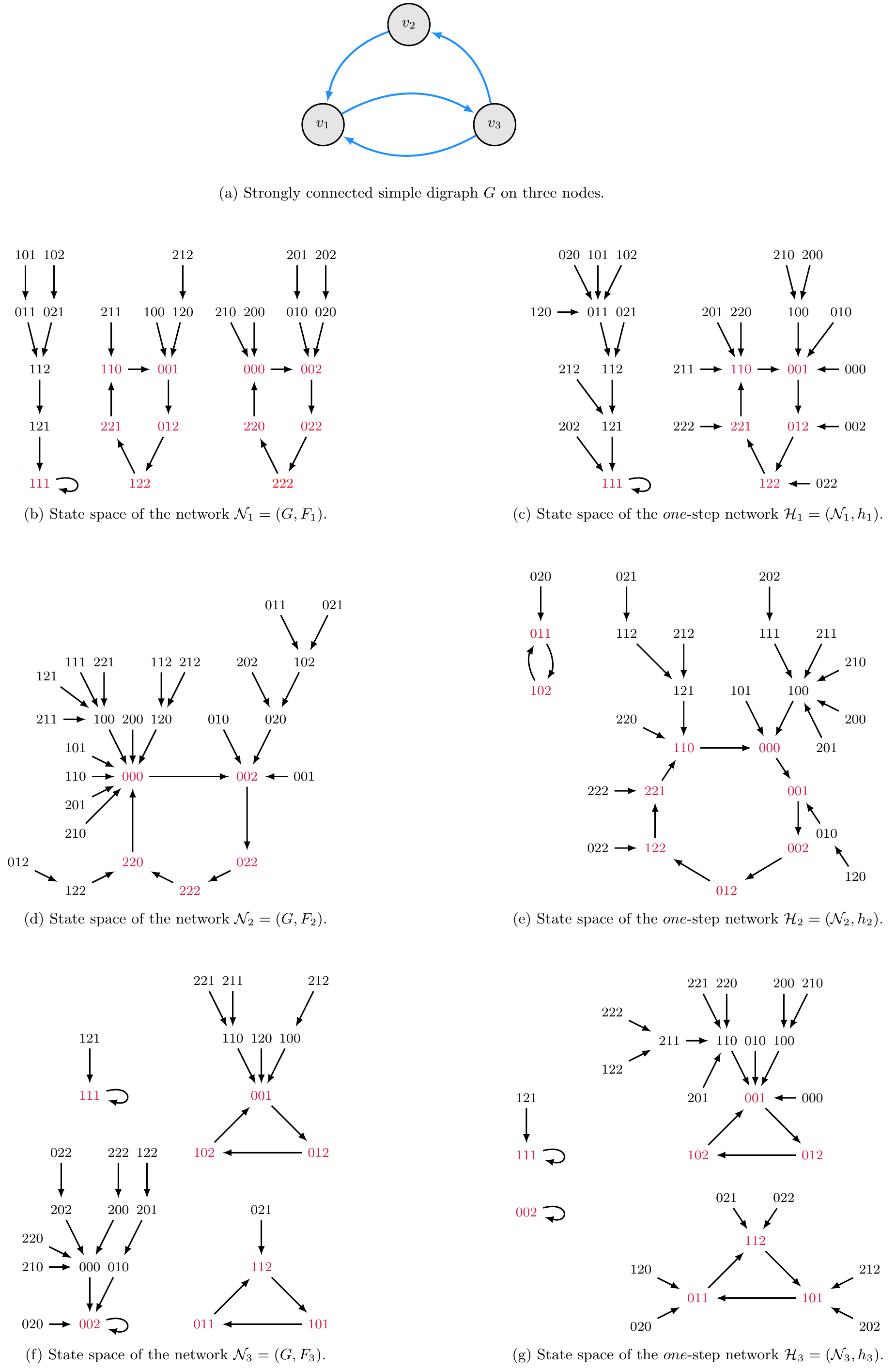}
\caption{Example of ternary network. Panel (a) Interaction graph. Panels (b), (d) and (f) display  state space  dynamics for each network as described in Example~\ref{ex:networks}. Panels (c), (e) and (g) display dynamics of the corresponding $one$-step network; $one$-step networks defined and described in Section~\ref{sec:one-step}. Fixed points and cycle attractors are colored in red. For simplicity we represent each state  $ (x_1,x_2,x_3)$ as $x_1x_2x_3$, e.g., the state $(1,1,2)$ is the same as $112$. 
}
\label{fig:1}       
\end{figure}

\noindent \textbf{Network $\mathcal{N}_1$:}  Consider a global transition function $F_1=(f_1, f_2, f_3): X_3^3 \to X_3^3$ given by 
\begingroup
\setlength{\abovedisplayskip}{4pt}
\setlength{\belowdisplayskip}{-9pt}
    \begin{align*}
    f_1(\mathbf{x}) &= \min(x_2, x_3)\\
    f_2(\mathbf{x}) &= x_3\\
    f_3(\mathbf{x}) &= 2-x_1.\\
    \end{align*}
\endgroup
Local update rules can also be represented by transition tables (Table~\ref{tab:n1}). Notice that each local activation function will either increase or decrease the future value of its target node.  Thus, with this definition of a global transition function, the interaction graph $G$ is a signed digraph with arcs $(v_2,v_1),(v_3,v_2)$ and $(v_3,v_1)$ being positive $(+)$, while the arc $(v_1,v_3)$ is negative $(-)$. The network  $\mathcal{N}_1=(G, F_1)$ has one fixed point and  two 5-cycles (Figure~\ref{fig:1}(b)). 
\vspace{-0.07in}

\begin{table}[h]
\caption{Transition tables representing local activation functions.}
\label{tab:n1}      
\begin{tabular}[t]{ll|l}
\hline\noalign{\smallskip}
$x_2$ & $x_3$ & $f_1 = \min(x_2, x_3)$  \\
\noalign{\smallskip}\hline\noalign{\smallskip}
0 & 0 & 0 \\
0 & 1 & 0 \\
0 & 2 & 0 \\
1 & 0 & 0 \\
1 & 1 & 1 \\
1 & 2 & 1 \\
2 & 0 & 0 \\
2 & 1 & 1 \\
2 & 2 & 2 \\
\noalign{\smallskip}\hline
\end{tabular}
\quad\quad
\begin{tabular}[t]{l|l}
\hline\noalign{\smallskip}
 $x_3$ & $f_2 = x_3$  \\
\noalign{\smallskip}\hline\noalign{\smallskip}
 0 & 0 \\
 1 & 1 \\
 2 & 2 \\
\noalign{\smallskip}\hline
\end{tabular}
\quad\quad
\begin{tabular}[t]{l|l}
\hline\noalign{\smallskip}
 $x_1$ & $f_3 = 2-x_1$  \\
\noalign{\smallskip}\hline\noalign{\smallskip}
 0 & 2 \\
 1 & 1 \\
 2 & 0 \\
\noalign{\smallskip}\hline
\end{tabular}
\end{table}
\vspace{-0.1in}

\noindent \textbf{Network $\mathcal{N}_2$:} Now consider a global transition function $F_2=(f_1, f_2, f_3): X_3^3 \to X_3^3$ given by 
\begingroup
\setlength{\abovedisplayskip}{4pt}
\setlength{\belowdisplayskip}{-9pt}
\begin{align*}
    f_1(\mathbf{x}) &= \min(x_2, x_3)\\
    f_2(\mathbf{x}) &= 2x_3+x_3^2 \bmod{3}\\
    f_3(\mathbf{x}) &= 2+x_1^2 \bmod{3}.\\
\end{align*}
\endgroup
In this case the interaction graph $G$ is also a signed digraph. The difference between network  $\mathcal{N}_1$  and $\mathcal{N}_2$ can be seen by comparing Table~\ref{tab:n1} and Table~\ref{tab:n2}. The network  $\mathcal{N}_2$ has one 5-cycle (Figure~\ref{fig:1}(d)).
\vspace{-0.07in}

\begin{table}[h!]
\caption{Transition tables representing local activation functions $f_2$ and $f_3$.}
\label{tab:n2}       
\begin{tabular}{l|l}
\hline\noalign{\smallskip}
 $x_3$ & $f_2 = 2x_3+x_3^2 \bmod{3}$  \\
\noalign{\smallskip}\hline\noalign{\smallskip}
 0 & 0 \\
 1 & 0 \\
 2 & 2 \\
\noalign{\smallskip}\hline
\end{tabular}
\quad\quad
\begin{tabular}{l|l}
\hline\noalign{\smallskip}
 $x_1$ & $f_3 = 2+x_1^2 \bmod{3}$  \\
\noalign{\smallskip}\hline\noalign{\smallskip}
 0 & 2 \\
 1 & 0 \\
 2 & 0 \\
\noalign{\smallskip}\hline
\end{tabular}
\end{table}
\vspace{-0.1in}

\noindent \textbf{Network $\mathcal{N}_3$:} Now consider a global transition function $F_3=(f_1, f_2, f_3): X_3^3 \to X_3^3$ given by
\begingroup
\setlength{\abovedisplayskip}{4pt}
\setlength{\belowdisplayskip}{-9pt}
    \begin{align*}
    f_1(\mathbf{x}) &= \min(x_2, x_3)\\
    f_2(\mathbf{x}) &= 2x_3+2x_3^2 \bmod{3}\\
    f_3(\mathbf{x}) &= 2-x_1.\\
    \end{align*}
\endgroup
In this case, the interaction graph $G$ is not a signed graph. The future value of the node $v_2$ first increases from 0 to 1 when the input node $v_3$ takes on the values of 0 and 1, respectively. However the value of the node $v_2$ then decreases back to 0 when the input node $v_3$ is $2$ (Table~\ref{tab:n3}). Thus, the arc from $v_3$ to $v_2$ cannot be assigned a unique sign. 
The network  $\mathcal{N}_3$ has two fixed points and two 3-cycles (Figure~\ref{fig:1}(f)).
\vspace{-0.07in}

\begin{table}[h!]
\caption{Transition table representing local activation function $f_2$.}
\label{tab:n3}       
\begin{tabular}{l|l}
\hline\noalign{\smallskip}
 $x_3$ & $f_2 = 2x_3+2x_3^2 \bmod{3}$  \\
\noalign{\smallskip}\hline\noalign{\smallskip}
 0 & 0 \\
 1 & 1 \\
 2 & 0 \\
\noalign{\smallskip}\hline
\end{tabular}
\end{table} 
\end{example}

\subsection{$one$-step $n$-ary Networks}\label{sec:one-step}

Dynamics in Figures~\ref{fig:1}(b), \ref{fig:1}(d) and \ref{fig:1}(f) demonstrate an important fact about $n$-ary networks when $n \geq 3$. A value of a node $v_i$ can increase, for example, from $0$ to $2$ in one time step and vice versa, e.g., $202 \rightarrow 020 \rightarrow 002 \rightarrow 022 \rightarrow$ etc. For $n=4$ the value can jump three units. In some applications it is desirable that the value of a node changes at most one unit in one time step (e.g., \cite{chifman2017}).  This can be accomplished  by defining another function that depends on the current state of the node at time $t$ and its next value at time $t+1$, which is updated by a local activation function \cite{veliz2010}.

\begin{definition}\label{cont}
Let $\mathcal{N}= (G,F)$ be an $n$-ary network and let $f_v: X_n^m \to X_n$ be a local activation function for $v \in V(G)$. A {\em $one$-step} local activation function $h_v$ is defined by

\begin{equation}
  h_v(\mathbf{x})= 
  \begin{cases}
    x_v+1 &\text{if  $x_v < f_v(\mathbf{x})$} \\
    x_v &\text{if  $x_v = f_v(\mathbf{x})$}\\
    x_v-1 &\text{if  $x_v > f_v(\mathbf{x})$}
  \end{cases},
\end{equation}

where $\mathbf{x}  = (x_1, x_2, \dots, x_v, \dots, x_m) \in X_n^m$ is the input vector and $x_v \in X_n$ is the current state of the node $v$. 
A {\em $one$-step} global transition function is then \[h=(h_1, h_2, \dots, h_m) : X_n^m \to X_n^m.\] 
An $n$-ary $one$-step network will be denoted by $\mathcal{H} = (\mathcal{N}, h)$.
\end{definition}

\begin{example}
Consider network $\mathcal{N}_1$ as defined in Example~\ref{ex:networks}. Then 
\[F_1(2,1,0) = (f_1(2,1,0),f_2(2,1,0),f_3(2,1,0)) = (0,0,0).\] 
Now, since $f_1(2,1,0)=0$ and $x_1=2>0$, then $h_1(2,1,0) = x_1-1=1$. Similarly we compute $h_2$ and $h_3$. Thus, $h(2,1,0)=(1,0,0)$.
\end{example}

\begin{remark}
A $one$-step local activation function $h_v(\mathbf{x})$ can also be written in the following form: 
\begin{equation}\label{eq:sgn-form}
h_v(\mathbf{x}) = x_v + \sgn(f_v(\mathbf{x}) - x_v),
\end{equation}
where
\begin{equation*}
 \sgn(w)= 
  \begin{cases}
    1 &\text{if  $w>0$} \\
    0 &\text{if  $w=0$}\\
    -1 &\text{if  $w<0$}
  \end{cases}.
\end{equation*}

\end{remark}

Note that a $one$-step function $h_v(\mathbf{x}) = f_v(\mathbf{x})$ if and only if $\abs{x_v-f_v(\mathbf{x})}\leq1$.  Thus, for binary (2-ary) networks global transition functions $F$ and $h$ are identical. We would like to point out that  in \cite{veliz2010} and \cite{chifman2017} function $h_v$ was termed {\em continuous}. To avoid confusion between continuous networks in which dynamics are described by a system of differential equations, we have decided to use the terms {\em $one$-step function} and {\em $one$-step network} to describe systems that are discrete but require that each node change at most one unit in one time step. This terminology also allows for future investigations of $k$-step networks. 

\begin{remark}\label{rm:fixed_pt}
Fixed points of $\mathcal{N}= (G,F)$ and the corresponding $one$-step network $\mathcal{H} = (\mathcal{N}, h)$ are the same. Suppose $\mathbf{x} \in X_n^m$ is a fixed point of $\mathcal{N}$. Then $F(\mathbf{x}) = \mathbf{x}$. This means that $f_v(\mathbf{x}) = x_v$ for each local activation function, and consequently Definition~\ref{cont} implies that $h_v(\mathbf{x}) = x_v$. Conversely, if $\mathbf{x}$ is a fixed point of $\mathcal{H}$ then  $h(\mathbf{x}) =\mathbf{x}$ and $h_v(\mathbf{x}) = x_v$, which is only true if $x_v = f_v(\mathbf{x})$. Thus $\mathbf{x}$ is a fixed point of $\mathcal{N}$ as well. This is not the case for cyclic attractors.
\end{remark}

Figures~\ref{fig:1}(c), \ref{fig:1}(e) and \ref{fig:1}(g) display dynamics of the $one$-step networks. Notice, the 5-cycle of the network $\mathcal{N}_1$ with two step jumps in Figure~\ref{fig:1}(b) now has been absorbed by the basins of the fixed point and the other 5-cycle in the corresponding $one$-step network $\mathcal{H}_1$. The remaining 5-cycle and the fixed point in $\mathcal{H}_1$ are identical to those in $\mathcal{N}_1$. For $\mathcal{N}_3$ and $\mathcal{H}_3$, the attractors are the same, though a fixed point $002$ in $\mathcal{H}_3$ lost all states in its basin. 
One similarity between $\mathcal{H}_1$ and $\mathcal{H}_3$ is that no new attractors have been introduced. This is not the case for the $one$-step network $\mathcal{H}_2$, which has lost its 5-cycle and two new cycles have been introduced (a 2-cycle and a 7-cycle). In Section~\ref{sec:nonexpanding} we present conditions under which $one$-step networks do not add new cyclic attractors, but before we turn to our main results, we describe a subclass of $n$-ary networks.

\subsection{Generalized Boolean  $n$-ary Networks}\label{sub:gbn} 
A subclass of $n$-ary networks that we would like to single out is the one in which each local activation function $f_v$ is constructed from three operators 
\begin{equation}\label{eq:operators}
\max(\mathbf{x}), \quad \min(\mathbf{x}), \quad \text{and} \quad \mathrm{not}(x_i) = (n-1) - x_i.  
\end{equation} 
where $\mathbf{x}  = (x_1, x_2,  \dots, x_m) \in X_n^m$ and $ x_i \in X_n$. Our team utilizes this scheme for ternary protein-protein networks as it allows us to intuitively describe interactions among biological species (see \cite{chifman2017}). Notice that the $\mathrm{not}$ operator is self inverting, meaning $\mathrm{not}(\mathrm{not}(x_i))=x_i$.  If the node $v$ has only one incident node $u$ and no loop, then depending on the biological context, the local activation function $f_v$ can be defined as {\em logical identity} $f_v (\mathbf{x}) = x_u$ or as $f_v(\mathbf{x}) =(n-1)-x_u$.   The $\max$ and $\min$ operators are idempotent since $\max(x_i,x_i) = x_i$ and $\min(x_i,x_i) = x_i$ for all $x_i \in X_n$, which means that they are {\em semilattice operators} (see \cite{VELIZCUBA2019167} about semilattice networks). Some familiar properties that are satisfied by $\max, \min$ and $\mathrm{not}$ operators are listed in Appendix~\ref{ap:proofs_examples}.  Various {\em compositions} of $\min/\max/\mathrm{not}$ are also possible.  In the case when $n=2$, the operators $\max, \min$ and $\mathrm{not}$ are just the traditional $OR$, $AND$ and $NOT$ operators, respectively. This leads to a natural definition below.  

\begin{definition} Networks in which each local activation function is a {\em composition} of the operators $\{\min,  \max, \mathrm{not}\}$ given by Equation~\ref{eq:operators} will be called {\bf generalized Boolean $n$-ary networks}, denoted by $\mathcal{N}_{GB}$.  
\end{definition}
An example of the $\mathcal{N}_{GB}$ 3-ary (ternary) network is $\mathcal{N}_1$ in Example~\ref{ex:networks}. Networks $\mathcal{N}_2$ and $\mathcal{N}_3$ are not generalized Boolean. Our next definition is an analogue of the disjunctive normal form for Boolean expressions with the difference being that  we allow for both $x_i$ and $\mathrm{not}(x_i)$ to appear in some terms. In the Boolean formalism, 
$\min(x, \mathrm{not}(x)) = 0$ for $x \in \{0,1\}$. This identity does not hold when $x \in \{0,1,2, \dots, n-1\}$ with $n>2$, unless $x=0$ or $x=n-1$. 

\begin{definition} Let  $\mathcal{N}_{GB}$ be a generalized Boolean $n$-ary network and $\mathbf{x} \in X_n^m$. A local activation function $f_v : X_n^m \to X_n$ is in \textbf{max--min form} if it is of the form 
\begin{equation}\label{eq:max_min}
f_v(\mathbf{x}) = \max(\mathcal{M}_1(\mathbf{x}), \mathcal{M}_2(\mathbf{x}), \dots, \mathcal{M}_r(\mathbf{x})), 
\end{equation}
\noindent where each $\mathcal{M}_i(\mathbf{x})= \min(a_{i1}, a_{i2}, \dots, a_{is})$, with $a_{ij} = x_k$ or $a_{ij} = \mathrm{not}(x_k)$, $x_k$ being one of the entries of the vector $\mathbf{x}$.    
In this context, each $\mathcal{M}_i$ is the $\min$ of some or all entries of the vector $\mathbf{x} = (x_1,x_2,\dots,x_n)$ with a possibility of having both $x_k$ and $\mathrm{not}(x_k)$ as part of the same $\mathcal{M}_i$ term. In the case $r=1$, $f_v(\mathbf{x}) = \mathcal{M}_1(\mathbf{x})$.\end{definition}

\begin{lemma}\label{lem:max_min}
Each local activation function $f_v$ of a generalized Boolean $n$-ary network can be written in $\max-\min$ form. 
\end{lemma}
\begin{proof}
Repeated application of the De Morgan's laws \eqref{subeq:e} - \eqref{subeq:f} and the Distributive property \eqref{subeq:c} (see Appendix~\ref{ap:proofs_examples}) will transform $f_v$ into $\max-\min$ form. Identical terms can be eliminated by using the fact that the $\max$ operator is idempotent. 
\end{proof}

\begin{example}\label{ex:min_max} 
Consider $f_v: X_n^3 \to X_n$ defined by \[f_v(\mathbf{x}) = \min(x_2, \max(x_1,\mathrm{not}(\min(x_2,x_3)))).\] We are going to transform $f_v$ into $\max - \min$ form. Note that 
\begin{equation}\label{eq:maxrule}
\max(x,y,z) = \max(\max(x,y),z) = \max(x,\max(y,z)).
\end{equation} The same is true for the $\min$ operator. Now,
\begin{align*}
&\min(x_2, \max(x_1,\mathrm{not}(\min(x_2,x_3))))\\ 	&\stackrel{\eqref{subeq:c}}{=} \max(\min(x_2, x_1), \min(x_2, \mathrm{not}(\min(x_2,x_3))))\\
										&\stackrel{\eqref{subeq:f}}{=}\max(\min(x_2, x_1), \min(x_2, \max(\mathrm{not}(x_2),\mathrm{not}(x_3))))\\
										&\stackrel{\eqref{subeq:c}}{=}\max(\min(x_2, x_1), \max(\min(x_2, \mathrm{not}(x_2)),\min(x_2,\mathrm{not}(x_3))))\\
										&\stackrel{\text{Eq.}\ref{eq:maxrule}}{=}\max(\min(x_1, x_2), \min(x_2, \mathrm{not}(x_2)),\min(x_2,\mathrm{not}(x_3))).				
\end{align*}
Thus, $f_v(\mathbf{x}) = \max(\mathcal{M}_1(\mathbf{x}), \mathcal{M}_2(\mathbf{x}), \mathcal{M}_3(\mathbf{x}))$, where $\mathcal{M}_1(\mathbf{x})~=~\min(x_1, x_2)$, $\mathcal{M}_2(\mathbf{x})~=~\min(x_2, \mathrm{not}(x_2))$ and 
$\mathcal{M}_3(\mathbf{x})~=~\min(x_2,\mathrm{not}(x_3))$. 
\end{example}

\section{Nonexpanding Networks}\label{sec:nonexpanding}
The main objective of this section is to establish conditions under which no new cyclic attractors  are introduced in $one$-step networks. It is quite clear from Figure~\ref{fig:1} that a $one$-step global transition function does change trajectories. It is not quite obvious, though, if there is a connection between the trajectories in Figure~\ref{fig:1}(b) and Figure~\ref{fig:1}(c), for example.  We initially speculated that a global transition function of the $one$-step network acts ``chaotically'' on a trajectory of an initial vector. To our surprise, we found that if a global transition function $F$ has some desirable properties then the corresponding $one$-step global transition function preserves some rigid structure, allowing us to establish our main results.  To make these ideas concrete we begin with a few examples.   

Consider networks $\mathcal{N}_1$ and $\mathcal{H}_1$ in Example~\ref{ex:networks} and its corresponding state spaces in Figure~\ref{fig:1}(b) and Figure~\ref{fig:1}(c). Let $\mathbf{x} = (2,0,0)$ be an initial state, then 
\[
\begin{matrix}
\text{trajectory under $\mathcal{H}_1$:} 	& 200 	& \rightarrow 	& 100 		&\rightarrow 	&001 		&\rightarrow 	& 012 			&\rightarrow  	&\cdots\\
				& \textbf{x} &			& h(\textbf{x}) 	&			& h^2(\textbf{x}) &			& h^3(\textbf{x}) 	&			&\cdots \\
				&  &			& h(\textbf{x}) 	&			& F(h(\textbf{x})) &			& F^2(h(\textbf{x})) 	&			&\cdots
\end{matrix}
\]
One can see that when $h(\mathbf{x}) = h(2,0,0) = (1,0,0)$ is computed, the trajectory from the point $h(\mathbf{x})=(1,0,0)$ is the same for both $\mathcal{N}_1$ and $\mathcal{H}_1$ networks.  This is true for all initial states in $\mathcal{H}_1$. One also checks that the same phenomenon occurs for $\mathcal{N}_3$ and $\mathcal{H}_3$ networks even though this network is not generalized Boolean and the interaction graph is not signed. 

Now, consider network $\mathcal{N}_2$ and $\mathcal{H}_2$ (state spaces are in Figure~\ref{fig:1}(d) and Figure~\ref{fig:1}(e)) and let $\mathbf{x} = (2,1,2)$.  Then 
 \[
\begin{matrix}
\text{trajectory under $\mathcal{H}_2$:} & 212 & \rightarrow & 121 &\rightarrow & 110 &\rightarrow & 000 &\rightarrow  	&\cdots\\
			& \textbf{x} && h(\textbf{x}) && h^2(\textbf{x}) && h^3(\textbf{x}) &&\cdots 
\end{matrix}
\]

\[
\begin{matrix}
\text{trajectory under $\mathcal{N}_2$:} & 121 &\rightarrow & 100 &\rightarrow & 000 &\rightarrow & 002 &\rightarrow  	&\cdots\\
			& h(\textbf{x}) && F(h(\textbf{x})) && F^2(h(\textbf{x})) && F^3(h(\textbf{x})) &&\cdots
\end{matrix}
\]
The trajectories starting from $h(\textbf{x}) = h(2,1,2) = (1,2,1)$ are quite different and while they do coincide at some states they never merge, meaning there is no state from which the trajectory is the same for both $\mathcal{N}_2$ and $\mathcal{H}_2$ networks.  

To study this phenomenon, as demonstrated by examples above, we need to be able to measure a distance between a pair of points $\mathbf{x}, \mathbf{y} \in X_n^m$. The choice of a distance below is influenced by Definition~\ref{cont} and its implication: a $one$-step local activation function $h_v(\mathbf{x}) = f_v(\mathbf{x})$ if and only if $\abs{x_v-f_v(\mathbf{x})}\leq1$, where $x_v$ is the current state of the node $v$ and $f_v$ is the local update function for $v$. Keeping this observation in mind, we use the maximum distance, also called the {\em Chebyshev distance}, between two elements $\mathbf{x}, \mathbf{y} \in X_n^m$ 
defined by \[d_\infty(\mathbf{x}, \mathbf{y}) = \max_{v \leq m} \abs{x_v-y_v}.\] 
This now leads to a natural definition of closeness. 

\begin{definition}\label{defn:close}
Let $\mathcal{N} = (G,F)$ be an $n$-ary network  and let $\mathbf{x}, \mathbf{y} \in X_n^m$. 
\begin{enumerate}
\item[(i)] A point $\mathbf{x}$ is {\bf close} to a point $\mathbf{y}$ if $d_\infty(\mathbf{x},\mathbf{y}) \leq 1$. 
\item[(ii)] A global transition function $F$ is {\bf nonexpanding} 
if for all $\mathbf{x}, \mathbf{y} \in X_n^m$ such that $d_\infty(\mathbf{x}, \mathbf{y}) \leq 1$ we have $d_\infty(F(\mathbf{x}), F(\mathbf{y})) \leq 1$. In this case we also say that an $n$-ary network $\mathcal{N}$ is a {\bf nonexpanding} network. 
\end{enumerate}
\end{definition}

The term {\em nonexpanding} is motivated by the fact that in nonexpanding networks a global transition function is a discrete analogue of a mapping satisfying the Lipschitz condition with a constant equal to one. Such mappings are called {\em nonexpansive} or {\em nonexpanding}.
Also note,  Boolean (2-ary) networks are necessarily nonexpanding.

\begin{lemma}\label{lem:nonexp}
Let $\mathcal{N} = (G,F)$ be an $n$-ary network. $\mathcal{N}$ is nonexpanding if and only if $
d_\infty(F(\mathbf{x}),F(\mathbf{y}))~\leq~d_\infty(\mathbf{x}, \mathbf{y})$ for all $\mathbf{x}, \mathbf{y} \in X_n^m$.  
\end{lemma}
\begin{proof}
If $d_\infty(F(\mathbf{x}),F(\mathbf{y}))~\leq~d_\infty(\mathbf{x}, \mathbf{y})$ for all $\mathbf{x}, \mathbf{y} \in X_n^m$ then $d_\infty(F(\mathbf{x}),F(\mathbf{y}))~\leq~1$ whenever $d_\infty(\mathbf{x}, \mathbf{y})~\leq~1$. 
For the other direction, suppose $\mathcal{N}$ is nonexpanding and let $\mathbf{x}, \mathbf{y} \in X_n^m$.  If $\mathbf{x} = \mathbf{y}$ then the inequality holds trivially. 

Let $d_\infty(\mathbf{x}, \mathbf{y}) = k$ for some $k \in X_n$. This means that each coordinate $x_i$ of $\mathbf{x}$ differs from the coordinate $y_i$ of the vector $\mathbf{y}$ by at most $k$ units. Thus, we can find $k-1$ vectors $\mathbf{z}_j \in X_n^m$ such that $d_\infty(\mathbf{x}, \mathbf{z}_1)~\leq~1$, $d_\infty(\mathbf{z}_{k-1}, \mathbf{y})~\leq~1$ and  $d_\infty(\mathbf{z}_i, \mathbf{z}_{i+1})~\leq~1$ for $i \in \{1, 2, \dots, k-2\}$.   Now using the triangle inequality and the assumption that $\mathcal{N}$ is nonexpanding we obtain

 \begin{align*}
d_\infty(F(\mathbf{x}),F(\mathbf{y})) &\leq d_\infty(F(\mathbf{x}),F(\mathbf{z}_1)) + d_\infty(F(\mathbf{z}_{k-1}),F(\mathbf{y})) + \sum_{i=1}^{k-2} d_\infty(F(\mathbf{z}_i),F(\mathbf{z}_{i+1}))\\
& \leq \underbrace{1 + 1+ \cdots + 1}_\text{$k$ times} \\
& = k\\
& = d_\infty(\mathbf{x}, \mathbf{y}). 
\end{align*}
\end{proof}

Networks $\mathcal{N}_1$ and $\mathcal{N}_3$ in Example~\ref{ex:networks} are nonexpanding, though $\mathcal{N}_3$ is not a generalized  Boolean network. The network $\mathcal{N}_2$ fails to be nonexpanding since $d_\infty((1,1,2),(0,0,1))=1$ but $d_\infty(F_2(1,1,2),F_2(0,0,1)) =2.$

Our first main result of this section concerns generalized Boolean networks. 


\begin{theorem}\label{thm:gen_boolean}
A generalized Boolean $n$-ary network $\mathcal{N}_{GB}$ is nonexpanding.
\end{theorem}
\begin{proof}
Let $\mathbf{x}, \mathbf{y} \in X_n^m$ so that $d_\infty(\mathbf{x}, \mathbf{y}) \leq 1$. Since 
\[d_\infty(F(\mathbf{x}), F(\mathbf{y})) = \max_{v \leq m} \abs{f_v(\mathbf{x}) - f_v(\mathbf{y})},\] 
we only need to show that $\abs{f_v(\mathbf{x}) - f_v(\mathbf{y})}~\leq~1$ for each local activation function. By Lemma~\ref{lem:max_min}, 
$f_v$ of the generalized Boolean $n$-ary network can be written in $\max-\min$ form. Thus, $f_v(\mathbf{x})$ and $f_v(\mathbf{y})$ can be expressed in the form of Equation~\ref{eq:max_min} 
where $\mathcal{M}_i(\mathbf{x}) = \min(a_{i1}, a_{i2}, \dots, a_{is})$ and $\mathcal{M}_i(\mathbf{y}) = \min(b_{i1}, b_{i2}, \dots, b_{is})$ with $a_{ij}~=~x_k$ or $a_{ij}~=~\mathrm{not}(x_k)$ and $b_{ij}~=~y_k$ or $b_{ij}~=~\mathrm{not}(y_k)$. Now, $x_k$ and $y_k$ are some entries of the the vectors $\mathbf{x}$ and $\mathbf{y}$, respectively, corresponding to the same index $k$.  
Since $\mathbf{x}$ is close to $\mathbf{y}$, $y_k$ differs from $x_k$ by at most one unit.   
Clearly, $\mathrm{not}(x_k)$ and $\mathrm{not}(y_k)$ also differ by at most one unit. 
Consequently, for all $j \in \{1, \dots, s\}$
\begin{equation}\label{eq:abs}
\abs{b_{ij} - a_{ij}}~\leq~1.
\end{equation}  

Without loss of generality, assume that $a_{i1} \leq a_{ij}$ for all $j \in \{1, \dots, s\}$. That is $a_{i1}~=~\min(a_{i1}, a_{i2}, \dots, a_{is})$. Suppose  $b_{ij}~ =~\min(b_{i1}, b_{i2}, \dots, b_{is})$ for some $j \in \{1, \dots, s\}$.  Then by Eq.~\ref{eq:abs} and the fact that $a_{i1}$ and $b_{ij}$ are the minima, we get
\begin{equation*}
a_{i1}-1 \leq a_{ij}-1 \leq b_{ij} \leq b_{i1} \leq a_{i1}+1.
\end{equation*}
Thus we conclude, $\abs{\mathcal{M}_i(\mathbf{x}) - \mathcal{M}_i(\mathbf{y})} \leq1$ for all $i \in \{1, \dots, r\}$. 

We can apply similar logic to the $\max$ operator to show that for all $v \in \{1, \dots, m\}$
\[\abs{f_v(\mathbf{x}) -f_v(\mathbf{y})} =\abs{\max(\mathcal{M}_1(\mathbf{x}), \dots, \mathcal{M}_r(\mathbf{x})) - \max(\mathcal{M}_1(\mathbf{y}), \dots, \mathcal{M}_r(\mathbf{y}))} \leq1.\] 

Hence, $d_\infty(F(\mathbf{x}), F(\mathbf{y})) = \max_{v \leq m} \abs{f_v(\mathbf{x}) - f_v(\mathbf{y})} \leq 1.$

\end{proof}

\subsection{$one$-step Networks and Cyclic Attractors}
As we have already stated, networks $\mathcal{N}_1$ and $\mathcal{N}_3$ in Example~\ref{ex:networks} are nonexpanding, while the network $\mathcal{N}_2$ is not (e.g., close points $(1,1,1)$ and $(0,0,1)$ are mapped by $F_2$ to $(1,0,0)$ and $(0,0,2)$, respectively). The corresponding $one$-step networks also have the same feature, meaning that $\mathcal{H}_1$ and $\mathcal{H}_3$ are nonexpanding, while $\mathcal{H}_2$ is not. It could also happen that a network that fails to be nonexpanding has a corresponding $one$-step network that maps all close points to close points. For example, a global transition function $F = (f_1,f_2): X_3^2 \to X_3^2$ of the ternary network on two nodes (Equation~\ref{eq:two_node} below), will not map all close points to close points, however, the corresponding $one$-step network will (state space for this example is provided in Appendix~\ref{ap:proofs_examples}: Example~\ref{ex:4} and Figure~\ref{fig:3})
\begin{align}\label{eq:two_node}
f_1(x_1,x_2) &= (2+x_1+x_2+2x_1 x_2 +x_1^2 x_2^2)\bmod{3}\\
f_2(x_1,x_2) &= x_1.\nonumber
\end{align}  
The lemma below summarizes this discussion. A proof can be found in Appendix~\ref{ap:proofs_examples}.
 \begin{lemma}\label{lem:nonex_h}
 Suppose $\mathcal{N} = (G,F)$ is an $n$-ary network with the corresponding $one$-step network $\mathcal{H} = (\mathcal{N},h)$. If $\mathcal{N}$ is nonexpanding, then so is $\mathcal{H}$.
 \end{lemma}

Another important observation that will be used in our proofs follows from Definition~\ref{defn:close} and Section~\ref{sec:one-step}. Given any $n$~-~ary network $\mathcal{N} = (G,F)$ and  a corresponding $one$-step network $\mathcal{H} = (\mathcal{N},h)$, $F(\mathbf{x}) = h(\mathbf{x})$  if and only if $F(\mathbf{x})$ is {\em close} to $\mathbf{x}$.  Additionally, $h^{k-1}$ is always {\em close} to $h^k$, that is $d_\infty(h^{k}(\mathbf{x}),h^{k-1}(\mathbf{x}))~\leq~1$, where $h^i$ is the composition of $h$ with itself $i$ times.  
 
Before we introduce and prove our main results we will need the following lemma. In Appendix~\ref{ap:proofs_examples} we introduce a few other small results.
   
 \begin{lemma}\label{lem:properties}
 Suppose $\mathcal{N} = (G,F)$ is a nonexpanding $n$-ary network with the corresponding $one$-step network $\mathcal{H} = (\mathcal{N},h)$. Let $\mathbf{x} \in X_n^m$ and $i \in \mathbb{N}$. If $F(\mathbf{x})~=~h(\mathbf{x})$ then $F^i(\mathbf{x})~=~h^i(\mathbf{x})$.
\end{lemma}
\begin{proof}
Note that $F(\mathbf{x}) = h(\mathbf{x})$ if and only if $d_\infty(F(\mathbf{x}), \mathbf{x}) \leq1$. When $i=1$, the base case holds by assumption. Now suppose $F^{i-1}(\mathbf{x}) = h^{i-1}(\mathbf{x})$. 
Since $F$ is nonexpanding, we can apply Lemma~\ref{lem:nonexp} repeatedly to obtain desired result:
\begin{align*}
d_\infty(F^{i}(\mathbf{x}), F^{i-1}(\mathbf{x})) 	&= d_\infty(F(F^{i-1}(\mathbf{x})), F(F^{i-2}(\mathbf{x})))\\
										&\leq d_\infty(F^{i-1}(\mathbf{x}), F^{i-2}(\mathbf{x}))\\
										&\leq d_\infty(F^{i-2}(\mathbf{x}), F^{i-3}(\mathbf{x}))\\
										& \vdots \\
										&\leq d_\infty(F(\mathbf{x}), \mathbf{x})\\
										&\leq 1.
\end{align*}

Thus, $F^{i}(\mathbf{x}) = h(F^{i-1}(\mathbf{x})) = h(h^{i-1}(\mathbf{x})) = h^i(\mathbf{x})$.
\end{proof}
 
Our next Theorem addresses the discussion at the beginning of Section~\ref{sec:nonexpanding} about trajectories. The vital observation we have made about trajectories of nonexpanding networks is that they eventually merge. The maximum number of compositions a $one$~-~step function $h$ must make before the trajectory of the $one$-step network $\mathcal{H}$ becomes the same as for the network $\mathcal{N}$, depends on the number of states $n$ but not the number of nodes m. 

\begin{theorem}\label{thm:composition}
Suppose $\mathcal{N} = (G,F)$ is an $n$-ary nonexpanding network with the corresponding $one$-step network $\mathcal{H} = (\mathcal{N},h)$. Then, for all $\mathbf{x} \in X_n^m$ and $i \in \mathbb{N}$ \[F^i(h^{n-2}(\mathbf{x})) =h^i( h^{n-2}(\mathbf{x})).\]
\end{theorem}
\begin{proof}
First we show that for $0 \leq i\leq n-2$ and $\mathbf{x} \in X_n^m$ 
\begin{equation*}
d_\infty(F(h^i(\mathbf{x})), h(h^i(\mathbf{x}))) = \max_{v \leq m} \abs{f_v(h^i(\mathbf{x})) - h_v(h^i(\mathbf{x}))}
								      \leq	(n-2) - i. 
\end{equation*}
In particular, we show the above inequality holds for each local activation function.  
Recall, $\abs{f_v ({h^{i}(\mathbf{x})}) - h_v (h^{i-1}(\mathbf{x}))}~\leq~1$ if and only if $f_v ({h^{i}(\mathbf{x})})~=~h_v (h^{i}(\mathbf{x}))$. 
\\

\noindent Thus, we may assume that  $f_v ({h^{i}(\mathbf{x})})~>~h_v (h^{i}(\mathbf{x}))$. The other inequality follows analogously. With this assumption, 
\begin{equation}\label{eq:1}
h_v (h^{i}(\mathbf{x})) = h_v (h^{i-1}(\mathbf{x})) +1 \ \text{and} \ f_v ({h^{i}(\mathbf{x})}) - h_v (h^{i-1}(\mathbf{x})) > 1.
\end{equation}
Additionally, since $h^{i-1}(\mathbf{x})$ is close to $h^{i}(\mathbf{x})$ and $F$ is nonexpanding, we get
\begin{align*}
f_v(h^{i-1}(\mathbf{x}))+1 	&\geq f_v(h^i(\mathbf{x}))\\
					&> h_v(h^i(\mathbf{x}))\\
					&=h_v (h^{i-1}(\mathbf{x})) +1.
\end{align*}
Specifically, 
\begin{equation}\label{eq:2}
f_v(h^{i-1}(\mathbf{x})) > h_v (h^{i-1}(\mathbf{x})).
\end{equation}

Applying this argument repeatedly, we can show that Equations~(\ref{eq:1})~--~(\ref{eq:2}) hold for any $k~\in~\{0, 1, \dots, i\}$. Therefore, by the definition of a $one$-step function and ~{Remark~\ref{eq:sgn-form}},
\begin{align}\label{eq:3}
h_v(h^i(\mathbf{x})) 	&= x_v + \sum_{k=0}^i \sgn(f_v(h^k(\mathbf{x})) - h_v(h^{k-1}(\mathbf{x})))\nonumber \\
				&= x_v +(i+1),
\end{align}

where $h_v(h^{k-1}(\mathbf{x})) = x_v$ for $k=0$. Since the largest possible value for $f_v ({h^{i}(\mathbf{x})}) $ is $n-1$ and $h_v(h^i(\mathbf{x})) \geq i+1$ by Equation~(\ref{eq:3}), \[f_v ({h^{i}(\mathbf{x})}) - h_v (h^{i}(\mathbf{x})) \leq (n-1)-(i+1) = (n-2)-i.\]
Using similar logic for the other assumption, $f_v ({h^{i}(\mathbf{x})})~<~h_v (h^{i}(\mathbf{x}))$ we get
\[h_v ({h^{i}(\mathbf{x})}) - f_v (h^{i}(\mathbf{x})) \leq (n-1)-(i+1) = (n-2)-i.\]
Thus, we conclude that for all $0\leq i \leq n-2$ \[d_\infty(F(h^i(\mathbf{x})), h(h^i(\mathbf{x}))) \leq (n-2) - i.\] 
In particular, when $i~=~n-2$ then $d_\infty(F(h^{n-2}(\mathbf{x})), h(h^{n-2}(\mathbf{x}))) = 0$. This means that $F(h^{n-2}(\mathbf{x}))~=~h(h^{n-2}(\mathbf{x}))$. 

Now Lemma~\ref{lem:properties}~ implies \[F^i(h^{n-2}(\mathbf{x})) =h^i( h^{n-2}(\mathbf{x})).\]
\end{proof}

Equipped with Theorem~\ref{thm:composition}, we are now ready to prove our main objective of this article. In particular, we show that when a network is nonexpanding then  a cyclic attractor of the network $\mathcal{H}$ is also an attractor of the network $\mathcal{N}$, implying that no new cyclic attractors are introduced in $one$-step networks.   
 
\begin{theorem}\label{thm:no_cycles}
Let $\mathcal{N} = (G,F)$ be an $n$-ary nonexpanding network with the corresponding $one$-step network $\mathcal{H}~=~(\mathcal{N},h)$. If $\mathcal{C}$ is a $k$-cycle of the network $\mathcal{H}$ then $\mathcal{C}$ is a $k$-cycle of the network $\mathcal{N}$. 
\end{theorem}
\begin{proof} Let $\mathcal{C}=\{\mathbf{x}, h(\mathbf{x}), h^2(\mathbf{x}), \dots, h^{k-1}(\mathbf{x})\}$ be any $k$-cycle of the network $\mathcal{H}$, where $h^k(\mathbf{x}) = \mathbf{x}$ for $k \geq 1$ and $h^j(h^k(\mathbf{x})) \in \mathcal{C}$ for any $j \in \mathbb{N}$. 
Thus, if $k \leq n-2$ then $n-2 = k+j$, for some $j \in \mathbb{N}$, and \[h^{n-2}(\mathbf{x}) = h^{k+j}(\mathbf{x}) =h^j(h^k(\mathbf{x})) = h^j(\mathbf{x}) \in \mathcal{C}.\]  Also, if $k > n-2$ then trivially $h^{n-2}(\mathbf{x})  \in \mathcal{C}$. Thus, $h^i( h^{n-2}(\mathbf{x})) = h^k(\mathbf{x}) = \mathbf{x}$ for some $i$. Applying Theorem~\ref{thm:composition} repeatedly, we get
\begin{align*}
F^i(h^{n-2}(\mathbf{x})) = h^i( h^{n-2}(\mathbf{x})) &= \mathbf{x}\\
F^{i+1}(h^{n-2}(\mathbf{x})) = h^{i+1}( h^{n-2}(\mathbf{x})) &= h(\mathbf{x})\\
\vdots \\
F^{i+k}(h^{n-2}(\mathbf{x})) = h^{i+k}( h^{n-2}(\mathbf{x})) &= h^k(\mathbf{x}).
\end{align*}
Thus, $\mathcal{C}$ is also a $k$-cycle of $\mathcal{N}$.   
\end{proof}

\section{Application to Intracellular Iron Metabolism}\label{sec:application}
To demonstrate the utility of our analytical results, we use our previously published model of intracellular iron regulation specific to normal breast epithelial cells \cite{chifman2017}. This model dynamically links the iron core network (depicted in green in Figure~\ref{fig:2}) to iron utilization, oxidative stress response and oncogenic pathways. Briefly, almost all living organisms require iron for cellular respiration, oxygen transport, DNA synthesis and energy production. However, iron has an ability to exist in various oxidation states and can contribute to the formation of reactive oxygen species (ROS) that can damage DNA and other cellular structures.  Iron dysregulation can lead to iron overload or deficiency, both of which are detrimental to the organism. Additionally, altered iron metabolism, signified by the reduced intracellular iron export (Fpn) and increased  iron import (TfR1), has been well documented in tumors (see recent review by  Brown et al. \cite{Brown2020}). The model presented in \cite{chifman2017} was one of the first attempts to better understand the connection between intracellular iron metabolism and cancer.  

\begin{figure}[t]
\centering
  \includegraphics[scale=0.52]{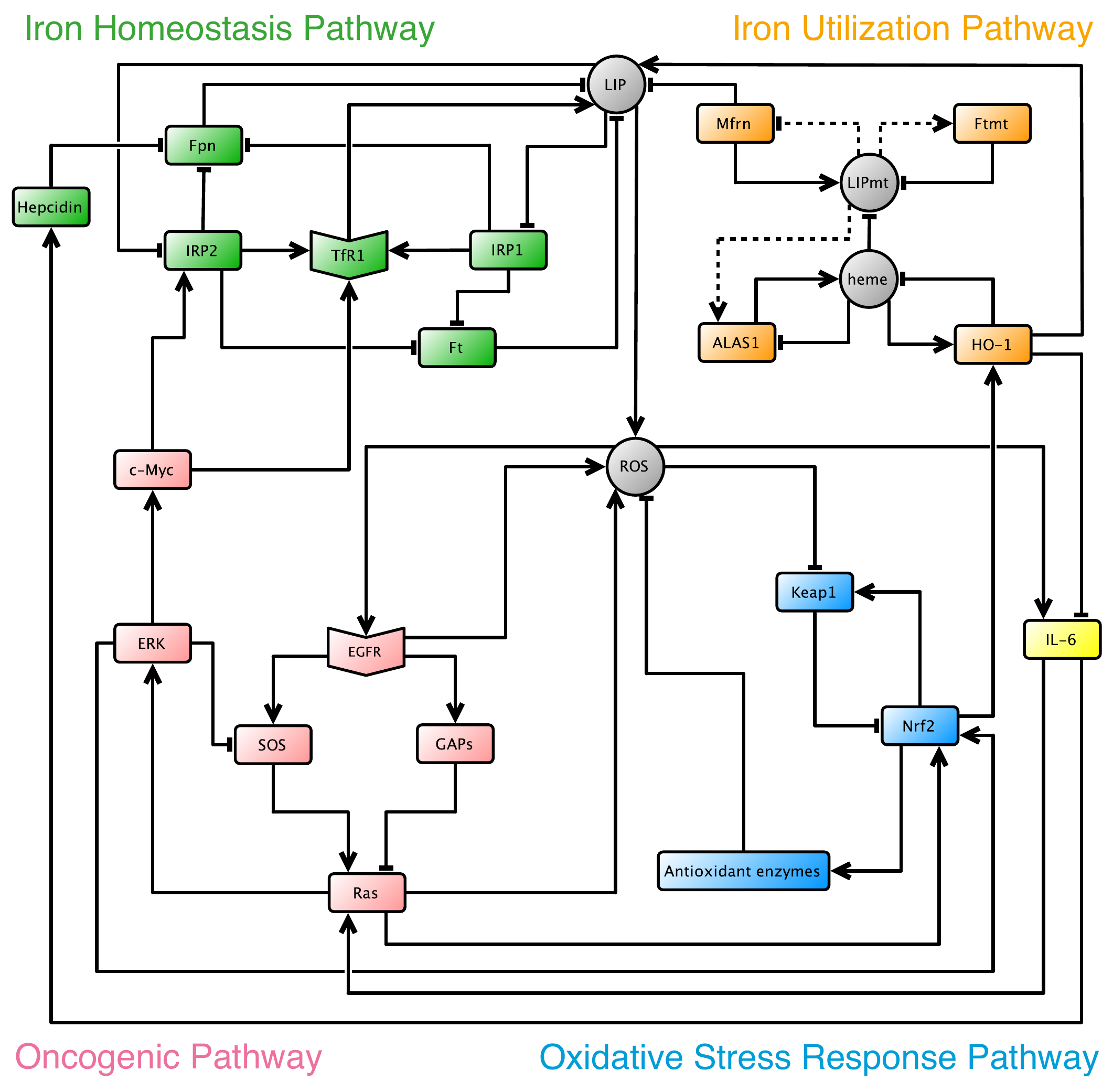}
\caption{Intracellular Iron Network. Original image as published in Chifman et al. (2017) \cite{chifman2017}. Activation/upregulation is represented by arrows and inhibition/downregulation by hammer heads.
}
\label{fig:2}       
\end{figure}
 For ease of notation, we use $x_i$ in place of the biological species, as listed below. See original article for the description of each variable and biological background in general.   
 
 \noindent\rule{\textwidth}{0.6pt} 
\begin{equation*}
\setcounter{MaxMatrixCols}{20}
\setlength\arraycolsep{1.2mm}
\begin{matrix}
\text{LIP} & \text{TfR1} & \text{Fpn} & \text{Ft} & \text{IRP1} & \text{IRP2} & \text{Hep} & \text{HO-1} & \text{ALAS1} & \text{Heme} & \text{ROS} & \text{AE} \\
x_1 & x_2 & x_3 & x_4 & x_5 & x_6 & x_7 & x_8 & x_9 & x_{10} & x_{11} & x_{12} \\

 & & & & & & & & & & & \\
\text{Nrf2} & \text{Keap1} & \text{IL-6} & \text{Ras} & \text{SOS} & \text{ERK} & \text{c-Myc} & \text{GAPs} & \text{EGFR} & \text{LIPmt} & \text{Mfrn} & \text{Ftmt}\\
x_{13} & x_{14} & x_{15} & x_{16} & x_{17} & x_{18} & x_{19} & x_{20} & x_{21} & x_{22} & x_{23} & x_{24}
\end{matrix}
\end{equation*}
 \noindent\rule{\textwidth}{0.6pt} \\
 
\noindent A global transition function $F: \{0,1,2\}^{24} \to \{0,1,2\}^{24}$ is given by
\begingroup
\allowdisplaybreaks
\begin{align*}
f_1 &= \min(\max(x_2,x_8),\min(\mathrm{not}(x_3),\mathrm{not}(x_4),\mathrm{not}(x_{23})))\\
f_2 &= \max(x_5^2 \bmod 3,x_6, x_{19})\\
f_3 &= \min(2+2x_5^2 \bmod 3,\mathrm{not}(x_6), \mathrm{not}(x_7))\\
f_4 &= \min(2+2x_5^2 \bmod 3,\mathrm{not}(x_6))\\
f_5 &= \mathrm{not}(x_1)\\
f_6 &= \max(\mathrm{not}(x_1),x_{19})\\
f_7 &= x_{15}\\
f_8 &= \max(x_{10},x_{13})\\
f_9 &= \min(\mathrm{not}(x_{10}),x_{22})\\
f_{10} &= \min(\mathrm{not}(x_8),x_9)\\
f_{11} &= \min(\max(x_1,x_{16},x_{21}),\mathrm{not}(x_{12}))\\
f_{12} &= x_{13}\\
f_{13} &= \max(\mathrm{not}(x_{14}),\max(x_{16},x_{18}))\\
f_{14} &= \min(\mathrm{not}(x_{11}),1+x_{13}+2x_{13}^2 \bmod 3)\\
f_{15} &= \max(\mathrm{not}(x_8),x_{11})\\
f_{16} &= \min(\max(x_{15},x_{17}),\mathrm{not}(x_{20}))\\
f_{17} &= \max(\mathrm{not}(x_{18}),x_{21})\\
f_{18} &= x_{16}\\
f_{19} &= x_{18}\\
f_{20} &= x_{21}\\
f_{21} &= x_{11}\\
f_{22} &= \min(x_{23},\min(\mathrm{not}(x_{10}),\mathrm{not}(x_{24})))\\
f_{23} &= \mathrm{not}(x_{22})\\
f_{24} &= x_{22}.
\end{align*}
\endgroup

This network is almost a generalized Boolean network with the exception of four local activation functions: $f_2, f_3, f_4$ and $f_{14}$. For example, the input nodes for $f_4$ are $x_5$ and $x_6$, but inside the $\min$ operator the variable $x_5$, which corresponds to the iron regulatory protein 1 (IRP1), appears in a polynomial function. The reason is that active IRP1 contributes less to the regulation of ferritin (Ft), the iron storage protein, meaning that when IRP1 = 2 it will have a weaker inhibitory impact than active IRP2. We termed such rules ``adjusted regulations'' in \cite{chifman2017} and they were necessary in order to adequately represent the strength of one biological species' control over another. 

We would like to argue that we can use Theorem~\ref{thm:gen_boolean} to conclude that this iron network is nonexpanding. Consider 
\begin{equation*}
f_4 = \min(2+2x_5^2 \bmod 3,\mathrm{not}(x_6))
\end{equation*}
and let $g(x_5) = 2+2x_5^2 \bmod 3$. Notice that $f_4$ is almost in $\max-\min$ form with one $\mathcal{M}_1(\mathbf{x})$ term:
\begingroup
\setlength{\abovedisplayskip}{7pt}
\setlength{\belowdisplayskip}{0pt}
\begin{equation*}
f_4 = \underbrace{\min(g(x_5),\mathrm{not}(x_6))}_\text{$\mathcal{M}_1(\mathbf{x})$}.
\end{equation*} 
\endgroup

The proof of Theorem~\ref{thm:gen_boolean} relies on the fact that for all close points $\mathbf{x}, \mathbf{y} \in X_n^m$ the corresponding entries $a_{ij}$ and $b_{ij}$ in $\mathcal{M}_i(\mathbf{x})~=~\min(a_{i1}, a_{i2}, \dots, a_{is})$ and $\mathcal{M}_i(\mathbf{y})~=~\min(b_{i1}, b_{i2}, \dots, b_{is})$ differ by at most one unit. Thus, all we need to check if for all $x_5, y_5 \in \{0,1,2\}$ such that  $\abs{x_5-y_5}~\leq~1$ we have $\abs{g(x_5)-g(y_5)}~\leq~1.$  Simple computation shows that this is true. Same logic and conclusions apply to $f_2, f_3$ and $f_{14}$. Therefore, we conclude that the iron network is nonexpanding and hence, by Theorem~\ref{thm:no_cycles} the corresponding $one$-step network does not have new cyclic attractors.   
 
When this iron network was first studied, having the information above would have been helpful. Our main goal in \cite{chifman2017} was to show that an oncogenic pathway alters the iron homeostasis pathway under a variety of experimental settings, such as knockout  or overexpression of critical species. 
Thus, we focused on attractors only.   We noticed early on that our original model included multiple cyclic attractors with $\pm 2$ jumps in addition to others, and we knew that running a corresponding $one$-step model would eliminate all cyclic attractors with jump discontinuities. However, we did not know if the corresponding $one$-step network would have new cyclic attractors at the time.  As a result, we had to run multiple models exhaustively. With current knowledge of nonexpanding networks, we could have extracted information from the original model about fixed points and cyclic attractors with $\pm 1$ jumps only, and then argued that under a $one$-step structure, the exact same attractors would be present, eliminating the need to run multiple models. 

\subsection{Iron Model that Fails to be Nonexpanding}
To demonstrate how quickly a network can fail the nonexpanding property, we can adjust, for example, the local activation function $f_{16}$ for Ras. The activity of Ras is controlled by a regulated GDP/GTP cycle, and its level of activation (GTP-bound~form~/~total~Ras) is regulated by the guanine nucleotide exchange factors and GTPase-activating  proteins (GAPs). Since the degree of Ras activation matters \cite{PMID:1549349,Boykevisch2006}, our original model considered three levels for Ras activation. On the other hand, if one is only concerned with Ras activity and not the levels, then Ras can be modeled as an on/off switch. One intuitive way to model this is to consider 0 as inactive and both $\{1,2\}$ as active, meaning when the activator of Ras, Son of Sevenless (SOS), is $\{1,2\}$ then Ras~=~2 (active), and when GAPs are $\{1,2\}$ then Ras~=~0 (inactive). The adjusted local activation function for Ras that reflects the above discussion is then
\begin{equation*}
f_{16} = \min(\max(x_{15}, 2 x_{17}^2 \bmod 3), 2+ x_{20}^2 \bmod 3).
\end{equation*}
Now, this network fails the nonexpanding property, since there are multiple close points $\mathbf{x}, \mathbf{y} \in \{0,1,2\}^{24}$ such that  
$\abs{f_{16}(\mathbf{x}) - f_{16}(\mathbf{y})}~=~2$. For example, one such pair of vectors would be obtained by setting all entries in  $\mathbf{x}$ to zero and all entries in $\mathbf{y}$ to zero except for $y_{17} = 1$. 

We have performed stochastic simulations with $10$-million random initial conditions and have found that the corresponding $one$-step model has at least three new cyclic attractors of length seven. To confirm that these three 7-cycles are not attractors without $one$-step structure, we nave traced a trajectory of one of the points from each 7-cycle. All trajectories converged to a different cyclic attractor.  Interestingly, the trajectories of vectors from each of the three cycles converged to the same discontinuous 7-cycle in the original network with modified local activation function for Ras. For simulations we have used our software~\url{https://steadycellphenotype.github.io}.

We have written a short R script that checks if an $n$-ary network is nonexpanding whenever Theorem~\ref{thm:gen_boolean} cannot be applied directly. We tested this script on the models discussed in this section and confirmed that the original model is nonexpanding, while the model with the  adjusted local activation function for Ras fails the property. The script is available at \url{https://github.com/chifman/Nonexpanding-networks}.

\section{Discussion}\label{sec:disc}
Networks in which variables take on $n>2$ discrete states have multiple scales of change for each variable, meaning that from one time step to the next a node can change by as much as $n-1$ units. To ensure that all nodes change by at most one unit, a $one$-step network was defined that depends on the structure and the global activation function of the original network. While these two networks are connected, the dynamics of their state spaces can be quite disparate. The results about cyclic attractors and trajectories of nonexpanding networks presented here are the first steps in understanding the state space dynamics between these networks. In particular, we have established that no new cyclic attractors are introduced in $one$-step networks, provided the original network was nonexpanding. Although the results in this article are specific to networks in which all variables take on the same number of discrete states and a synchronous update schedule, we hypothesize  that our results can be extended to other schemes, and in some cases the extensions might be trivial.    

Additionally, in this article we only considered the Chebyshev distance to define nonexpanding networks. The choice of this distance was motivated by the relation of the global function $F$ to its $one$-step function $h$. While other distances might lead to interesting conclusions and a class of networks with desirable properties, they do present subtle differences. Here, we briefly discuss the Taxicab distance (1~-~norm distance) and the Euclidean distance (2~-~norm distance) denoted by $d_1$ and $d_2$, respectively:
 \begin{equation}
d_1(\mathbf{x}, \mathbf{y}) = \sum_{v=1}^m \abs{x_v - y_v}, \quad d_2(\mathbf{x}, \mathbf{y}) = \left(\sum_{v=1}^m \abs{x_v - y_v}^2 \right)^{1/2}.
 \end{equation} 
 
We can modify the definition of ``closeness'' as follows. Let $\mathbf{x}, \mathbf{y}~\in~X_n^m$ and $i~\in~\{1,2\}$. We say that $\mathbf{x}$  is $d_i$~-~close to $\mathbf{y}$ if  $d_i(\mathbf{x}, \mathbf{y})~\leq~1$. Recall, $X_n~=~\{0,1,2,\dots,n-1\}$, thus this new definition implies that $\mathbf{x}$ and $\mathbf{y}$ differ by one unit in exactly one coordinate for both $d_1$ and $d_2$. Now we can say that a network is ``$d_i$-nonexpanding'' if $d_i(\mathbf{x}, \mathbf{y})~\leq~1$ implies $d_i(F(\mathbf{x}), F(\mathbf{y}))~\leq~1$ for all $\mathbf{x}, \mathbf{y}~\in~X_n^m$. One can show that if network $\mathcal{N}$ is $d_i$-nonexpanding then for all $\mathbf{x}, \mathbf{y}~\in~X_n^m$ 
\begin{align*}
d_1(F(\mathbf{x}), F(\mathbf{y})) &\leq d_1(\mathbf{x}, \mathbf{y})\\
d_2(F(\mathbf{x}), F(\mathbf{y})) &\leq \sqrt{m}\cdot d_2(\mathbf{x}, \mathbf{y}).
\end{align*}
However, if $\mathcal{N}$ is $d_i$-nonexpanding, the corresponding $one$-step network $\mathcal{H}$ may fail this property. For example, consider a ternary network $\mathcal{N}$ on two nodes given by 
\[F = (f_1, f_2) = (2+2x_1^2+x_1^2x_2^2 \bmod{3},x_2+2x_1^2x_2^2 \bmod{3}).\] This network is $d_i$-nonexpanding, but the corresponding $\mathcal{H}$ network is not.  For instance, let $\mathbf{x} = (0,2)$ and $\mathbf{y} = (1,2)$. Then 
\begin{align*}
d_i(\mathbf{x}, \mathbf{y}) &= 1 \quad \text{and}  \quad d_i(F(\mathbf{x}), F(\mathbf{y})) = d_i((2,2), (2,1)) = 1, \quad \text{but}\\
d_1(h(\mathbf{x}),h(\mathbf{y})) &= d_1((1,2), (2,1)) = 2\\
d_2(h(\mathbf{x}),h(\mathbf{y})) &= d_2((1,2), (2,1)) = \sqrt{2}.
\end{align*}  Additionally, this new definition does not guarantee that $h^{k-1}$ is {\em $d_i$-close} to $h^k$. Using the same example above, $d_1(\mathbf{y}, h(\mathbf{y}))~=~2$ and $d_2(\mathbf{y}, h(\mathbf{y}))~=~\sqrt{2}$. These properties played an important role in our proofs of Theorem~\ref{thm:composition} and Theorem~\ref{thm:no_cycles}. Moreover, we would like to point out that $d_i(F(\mathbf{x}), F(\mathbf{y}))$ may fail to be less than or equal to 1 for $d_i$~-~close points, and yet the network might have the property that no new cyclic attractors are added. Consider network $\mathcal{N}_1$ in Example~\ref{ex:networks}. Let $\mathbf{x}~=~(0,1,1)$ and $\mathbf{y}~=~(0,1,0)$. Then $d_i(\mathbf{x},\mathbf{y})~=~1$ but $d_1(F(\mathbf{x}), F(\mathbf{y}))~=~d_1((1,1,2), (0,0,2))~=~2~>~1$. Similarly, $d_2(F(\mathbf{x}), F(\mathbf{y}))~>~1$. 

We believe that networks in which $d_i(\mathbf{x}, \mathbf{y})~\leq~1$ implies $d_i(F(\mathbf{x}), F(\mathbf{y}))~\leq~1$ for $i~\in~\{1,2\}$ constitute a class of networks that has an intersection with the class of nonexpanding networks studied in this article in which no new cyclic attractors are added under the $one$-step function. However, it is not clear if $d_i$-closeness  is sufficient to guarantee the results of Theorem~\ref{thm:no_cycles}. For the time being we only pose this as a conjecture. 

These observations are important and will require careful investigation, which is the aim of our future work. We will also consider other definitions of closeness in an attempt to generalize our results and extend them to a wider class of networks. Another important question to consider is the stability of attractors in $one$-step networks, including fixed points. For example a basin of a fixed point can entirely disappear, raising a question about the importance of this particular fixed point ~{(see Figure~\ref{fig:1}~(f)~-~(g))}. Since fixed points between $one$-step networks and their original ones stay the same, it might be tempting to just compute fixed points and estimate their basins without passing to a $one$-step structure, missing the fact that this attractor might be biologically less meaningful if its basin can be dramatically reduced under a $one$-step function. We believe in the importance of pursuing these questions and providing analytical results, as they have the potential to be incorporated into modeling software and aid researchers  in their analyses of biological multi-state networks.

\section*{Author contributions}
Julia Chifman conceptualized, designed and directed the project. Original results for the nonexpanding {\em ternary networks} were derived by Etan Basser-Ravitz and Arman Darbar. All authors contributed to the extension of original results to multi-state networks. The first version of the manuscript was written by Julia Chifman and all authors commented and edited on subsequent versions of the manuscript. All authors read and approved the final manuscript.

\bibliographystyle{unsrtnat}

\newpage
\section*{Appendix}
\appendix
\section{Additional properties, proofs and examples}\label{ap:proofs_examples}

\textbf{Some properties of $\max, \min$ and $\mathrm{not}$ operators.}
\vspace{0.07in}

\noindent Let $x,y,z \in X_n=\{0,1,\dots,n-1\}$.

\begin{subequations}\label{eq:max_min_prop}
\noindent \text{Associative property:} 
\begin{align}
& \max(x, \max(y,z)) = \max(\max(x,y),z) \label{subeq:a}\\ 
& \min(x, \min(y,z)) = \min(\min(x,y),z).\label{subeq:b}
\end{align}
\text{Distributive property:} 
\begin{align}
& \min(x, \max(y,z)) = \max(\min(x,y),\min(x,z))\label{subeq:c}\\ 
& \max(x, \min(y,z)) = \min(\max(x,y),\max(x,z)).\label{subeq:d}
\end{align}
\text{De Morgan's laws:} 
\begin{align}
& \min(\mathrm{not}(x),\mathrm{not}(y)) = \mathrm{not}(\max(x,y))\label{subeq:e}\\
& \max(\mathrm{not}(x),\mathrm{not}(y)) = \mathrm{not}(\min(x,y)).\label{subeq:f}
\end{align}
\text{Absorption properties:}
\begin{align}
& \max(x, \min(x,y)) = x\label{subeq:g}\\
& \min(x, \max(x,y)) = x.\label{subeq:h}
\end{align}
\text{Other properties:}
\begin{align}
& \max(x, 0) =x\label{subeq:i}\\
& \min(x, n-1) = x\label{subeq:j}\\
& \mathrm{not}(\mathrm{not}(x)) = x. \label{subeq:k}
\end{align}
\end{subequations}
\vspace{0.1in}

\begin{proof}[\textbf{Proof of Lemma~\ref{lem:nonex_h}}]
Let $\mathbf{x}, \mathbf{y} \in X_n^m$ so that $d_\infty(\mathbf{x}, \mathbf{y}) \leq 1$.  
Since $\mathbf{x}$ is close to $\mathbf{y}$ then 
\begin{equation}\label{eq:close}
y_v = x_v-1 \quad \text{or} \quad y_v = x_v \quad \text{or}  \quad y_v = x_v+1,
\end{equation} 
where $x_v, y_v \in X_n$ are the current states of the node $v$. Let $f_v(\mathbf{x}) = k$ for some $k \in X_n$. This means that $f_v(\mathbf{y})$ takes on the values in the set $\{k-1, k, k+1\}$ because $F$ is nonexpanding. 

{\em Case 1.} Suppose $x_v < k$,  then according to Definition~\ref{cont} $h_v(\mathbf{x}) = x_v +1$. 

If $f_v(\mathbf{y}) = k-1$ then Definition~\ref{cont} and \ref{eq:close} together with the assumption that $x_v < k$ imply
\begin{equation*}
  h_v(\mathbf{y})= 
  \begin{cases}
    x_v &\text{if $y_v = x_v-1$} \\
    x_v +1 \ \text{or} \ x_v &\text{if $y_v = x_v$}\\
    x_v+2  \ \text{or} \ x_v+1 \ \text{or} \ x_v &\text{if $y_v = x_v +1$}.
  \end{cases}
\end{equation*}

If $f_v(\mathbf{y}) = k$ then using the same assumptions as above we get 
\begin{equation*}
  h_v(\mathbf{y})= 
  \begin{cases}
    x_v &\text{if $y_v = x_v-1$} \\
    x_v +1  &\text{if $y_v = x_v$}\\
    x_v+2  \ \text{or} \ x_v+1 &\text{if $y_v = x_v +1$}.
  \end{cases}
\end{equation*}

Finally, when $f_v(\mathbf{y}) = k+1$ then 
\begin{equation*}
  h_v(\mathbf{y})= 
  \begin{cases}
    x_v &\text{if $y_v = x_v-1$} \\
    x_v +1  &\text{if $y_v = x_v$}\\
    x_v+2  &\text{if $y_v = x_v +1$}.
  \end{cases}
\end{equation*}

Thus all possibilities lead to $\abs{h_v(\mathbf{x}) - h_v(\mathbf{y})} \leq 1$. 

{\em Case 2.} Suppose $x_v =k$,  then according to Definition~\ref{cont} $h_v(\mathbf{x}) = x_v $. Using similar logic as in Case~1, we get that the  possible values for $h_v(\mathbf{y})$ are $x_v-1, x_v$ or $x_v+1$, implying that $\abs{h_v(\mathbf{x}) - h_v(\mathbf{y})}~\leq~1$. 

{\em Case 3.} Suppose $x_v > k$,  then according to Definition~\ref{cont} $h_v(\mathbf{x}) = x_v -1$. Now, we get that the possible values for $h_v(\mathbf{y})$ are $x_v-2, x_v-1$ or $x_v$, which also leads to $\abs{h_v(\mathbf{x}) - h_v(\mathbf{y})} \leq 1$. 

In all cases, $d_\infty(h(\mathbf{x}), h(\mathbf{y})) = \max_{v \leq m} \abs{h_v(\mathbf{x}) - h_v(\mathbf{y})} \leq 1$. Thus, $\mathcal{H}$ is nonexpanding. 
\end{proof}

 \begin{lemma}\label{lem:dist1} Suppose $\mathcal{N} = (G,F)$ is an $n$-ary network with the corresponding $one$-step network 
 $\mathcal{H}~=~(\mathcal{N},h)$, then  $d_\infty(h(\mathbf{x}), F(\mathbf{x})) \leq n-2$ for all $\mathbf{x} \in X_n^m$. 
 \end{lemma}
 \begin{proof}
 
Since  $d_\infty(h(\mathbf{x}), F(\mathbf{x})) = \max_{v \leq m} \abs{h_v(\mathbf{x}) - f_v(\mathbf{x})}$, we only need to show that \\
$\abs{h_v(\mathbf{x}) - f_v(\mathbf{x})}~\leq~{n-2}$ for each local activation function.  Suppose $f_v(\mathbf{x})~=~k$ for some $k \in X_n$. Then according to Definition~\ref{cont} the values of $h_v$ will depend on the values of $x_v$, which is the current state of the node $v$. Note, if the value of $x_v$ were  $k-1, k$, or $k+1$, then $f_v (\mathbf{x})~=~h_v(\mathbf{x})$ and $\abs{h_v(\mathbf{x}) - f_v(\mathbf{x})}~=~0$. 

If $x_v \notin \{k-1, k, k+1\}$ then $h_v(\mathbf{x})~=~x_v+1$ if $x_v<k$ or $h_v(\mathbf{x})~=~x_v-1$ if $x_v>k$. 
Since $x_v, k \in X_n$,
the largest possible distance between $h_v(\mathbf{x})$ and $f_v(\mathbf{x})$ is $n-2$. In particular, 
$d_\infty(h(\mathbf{x}), F(\mathbf{x}))~\leq~n-2$.
 \end{proof}
 
  \begin{lemma}\label{lem:properties2}
 Suppose $\mathcal{N} = (G,F)$ is a nonexpanding $n$-ary network with the corresponding $one$-step network $\mathcal{H} = (\mathcal{N},h)$. Let $\mathbf{x} \in X_n^m$ and $i \in \mathbb{N}$. If $F(h^{i-1}(\mathbf{x}))~=~h(h^{i-1}(\mathbf{x}))$ then  $F(h^{i}(\mathbf{x}))~=~h(h^{i}(\mathbf{x}))$.
 \end{lemma}

\begin{proof}
Suppose $F(h^{i-1}(\mathbf{x})) = h(h^{i-1}(\mathbf{x})) = h^i(\mathbf{x})$. Since $h^{i-1}(\mathbf{x})$ is close to $h^{i}(\mathbf{x})$ and $F$ is nonexpanding, we get
\[
d_\infty(F(h^{i}(\mathbf{x})), h^{i}(\mathbf{x})) 	= d_\infty(F(h^{i}(\mathbf{x})), F(h^{i-1}(\mathbf{x}))) \leq 1.
\]
Thus we conclude that $F(h^{i}(\mathbf{x})) = h(h^{i}(\mathbf{x}))$.
\end{proof}

\begin{example}\label{ex:4}
Let $G$ be a digraph as depicted in Figure~\ref{fig:3}(a) and let $X_3 = \{0,1,2\}$. Let a global transition function $F=(f_1, f_2): X_3^2 \to X_3^2$ given by
    \begin{align*}
    f_1(\mathbf{x}) &= (2+x_1+x_2+2x_1 x_2 +x_1^2 x_2^2)\bmod{3}\\
    f_2(\mathbf{x}) &= x_1.\\
    \end{align*}
The network $\mathcal{N}$ fails to be nonexpanding and has one fixed point. The corresponding $one$-step network $\mathcal{H}$ is nonexpanding, but notice it now has a new 3-cycle. 
\end{example}

\begin{figure}[h!]
\centering
  \includegraphics[scale=0.75]{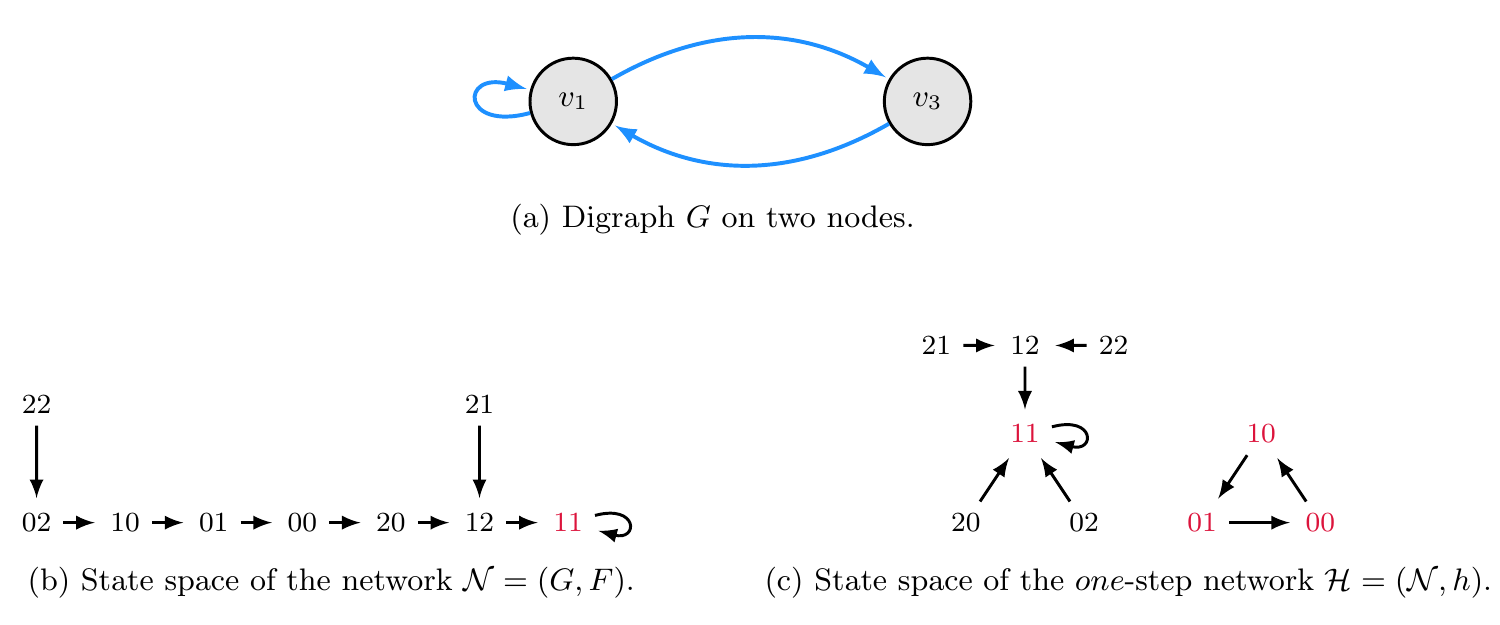}
\caption{Example of ternary network. Panel (a) Interaction graph. Panel (b) displays  state space  dynamics for network as described in Example~\ref{ex:4}. Panel (c) displays dynamics of the corresponding $one$-step network; $one$-step networks are described in the Section~\ref{sec:one-step}. Fixed points and cycle attractors are colored in red. 
}
\label{fig:3}      
\end{figure}

\end{document}